\documentclass[12pt]{article}
\usepackage{amsmath}
\usepackage{amssymb}
\usepackage{amsthm}
\usepackage{graphicx}
\usepackage[outdir=./]{epstopdf}
\usepackage{subcaption}
\usepackage{hyperref}
\usepackage{epstopdf}
\usepackage{datetime}
\usepackage{tabularx}
\newdate{date}{31}{05}{2018}
\date{\displaydate{date}}
\usepackage{multirow}
\newtheorem{remark}{Remark}
\usepackage[export]{adjustbox}
\newtheorem{assumption}{Assumption}

\newtheorem{theorem}{Theorem}

\usepackage{titling}
\usepackage{blindtext}
\usepackage[section]{placeins}
\usepackage{tikz}
\usepackage{pgfplots} 
\usepackage{pgfgantt}
\usepackage{pdflscape}
\pgfplotsset{compat=newest} 
\pgfplotsset{plot coordinates/math parser=false}
\newlength\figH
\newlength\figW

\usepackage{fancyhdr}
\usepackage{lipsum}

\title{LMI-Based Reset Unknown Input Observer for
State Estimation of Linear Uncertain Systems }
\author{Iman Hosseini$^{a}$, Alireza Khayatian$^{a}$, Paknoush Karimaghaee$^{a}$,\\ Mirko Fiacchini$^{b}$, Miguel Angel Davo$^{b}$,   \\
	\small $^{a}$Shiraz University, Shiraz, Iran \\
	\small $^{b}$Grenoble INP, GIPSA-lab, Grenoble, France \\
}
\providecommand{\keywords}[1]{\textbf{\textit{Keywords---}} #1}
\date{}

\begin{document}

	%############################################################
	
		\maketitle
		\thispagestyle{fancy}

%		\begin{titlingpage}	
		\begin{abstract}
		This paper proposes a novel kind of Unknown Input Observer (UIO) called Reset Unknown Input Observer (R-UIO) for state estimation of linear systems in the presence of disturbance using Linear Matrix Inequality (LMI) techniques. In R-UIO, the states of the observer are reset to the after-reset value based on an appropriate reset law in order to decrease the $L_2$ norm and settling time of estimation error. It is shown that the application of the reset theory to the UIOs in the LTI framework can significantly improve the transient response of the observer. Moreover, the devised approach can be applied to both SISO and MIMO systems. Furthermore, the stability and convergence analysis of the devised R-UIO is addressed. Finally, the	efficiency of the proposed method is demonstrated by simulation results.
		
		\end{abstract}
%	\end{titlingpage}

%	\blinddocument
%	\begin{titlepage}
%		\centering
%		\vfill
%		\huge{\textbf{Reset Unknown Input Observer}}\\
%		\vfill
%		By: Iman Hosseini\\
%		\vfill
%		%\today
%		\displaydate{date}		
%	\end{titlepage}
	%############################################################
		%############################################################
		\section{Introduction}
		The art of unknown input observer design for state estimation of systems with unknown inputs have attracted many attentions in the past decades.
		The problem of design a full-order observer for linear systems subjected to unknown inputs has been investigated in \cite{yang1988observers} and \cite{darouach1994full}. Besides, some research on the reduced order types can be found in \cite{kudva1980observers,guan1991novel,hou1992design,duan2001robust}. 		
		The existence of a UIO is investigated in \cite{darouach1994full,kudva1980observers,guan1991novel,hou1992design,chen1996design,chen2012robust,koenig2001design}, and the necessary and sufficient conditions are presented. Reduced order UIO can be designed using a systematic procedure. In this procedure, the state vector is partitioned into two parts using a linear transformation. Unknown Input directly affects one part and has to be measured completely and the reduced order UIO, which is decoupled from the input, is used to estimate the other part \cite{corless1998state}. The practical systems include unknown inputs such as the parameter perturbation \cite{park2014design}, actuator faults and external disturbance \cite{menon2014sliding} wherein the industrial process all can be viewed as UIs. Therefore, the discussions on UIO design are very important in both theory and applications, especially in the fields of observer-based control \cite{yang2014high,kim2016robust}. The performance of the UIOs in the presence of uncertainty and disturbance is outstanding \cite{hui2013stress,mondal2010lmi,xiong2003unknown}. Therefore, researchers have developed many different kinds of UIO.  In  \cite{chen2006unknown} Linear Matrix Inequalities are used to design  a full-order nonlinear UIO for a class of nonlinear Lipschitz systems with unknown input. Moreover, a reduced order UIO for one-sided nonlinear Lipschitz system is proposed in \cite{zhang2015unknown}. Considering uncertainties in the model, a robust UIO for fault detection using linear parameter varying model is investigated in \cite{li2017robust}.
		
		On the other hand, several control strategies are developed for dynamical systems  in the past decades. However, most of them suffer from having  oscillatory transient responses \cite{esfandiari2015stable,esfandiari2017adaptive}. In order to mitigate this issue and overcome the fundamental limitations of linear controllers, the idea of  reset control theory, in which a reset mechanism on the states of the controller is introduced, can be utilized. The idea of reset control originates from the Clegg Integrator (CI) which is aimed at overcoming the disadvantages of the traditional integrators. The state of CI is reset to zero when the input crosses zero \cite{clegg1958nonlinear}. This idea is extended to the first order element and the First Order Reset Element (FORE) is developed \cite{horowitz1975non}. However, answering the basic questions about well-posedness and stability has been the main problem in further developments of the reset control, for several years. In the late 1990s, the works of Chait, Hollot et al.  initiated a new session in reset control theory and gave a significant improvement in the field. They used the state space representation rather than transfer functions and started to answer such questions \cite{chait2002horowitz}. Since then on, this field has been an appealing research field and a number of international research groups have been working actively in this area and it turns to an attractive control design method with a significant potential for practical applications \cite{banos2011reset,beker2004fundamental,guo2012stability,davo2013reset}. 
		The stability problem and performance issue of the reset control systems is addressed in \cite{carrasco2013towards,banos2014network}. Based on the well-posedness of reset instants, the necessary and sufficient conditions for existence and uniqueness of solutions is developed in \cite{banos2016impulsive}. The problem of the global exponential stability of reset systems is discussed in \cite{fiacchini2016constructive}. A piecewise quadratic Lyapunov function is used in \cite{zhao2015reset} to deal with the stability of the reset control systems with uncertain outputs. In \cite{ghaffari2014reset}, a linear uncertain system with uncertain nonlinear terms is considered and a robust reset control law was designed but the stability at reset instances is not discussed. It is worth mentioning that all of the aforementioned systems can be viewed as nonlinear systems with the Lipschitz nonlinear term. In \cite{fiacchini2012quadratic} and \cite{fiacchiniHSC13}, systems with saturations and nested saturations is considered and quadratic and exponential stability is investigated in them  respectively.
		
		In the same way, a reset observer can be designed by applying reset mechanism into a traditional observer. A reset observer is a nonlinear observer consisting of a base observer and a reset law that resets the states of the observer when some predefined reset conditions are satisfied. The application of reset observers was first proposed in \cite{paesa2010reset}, in which a new type of adaptive observer is proposed. The designed observer is called Reset Adaptive Observer. In this observer, the integral term has been substituted by a reset element. In \cite{paesa2011optimal}, an optimal reset adaptive observer is designed, in which the observer parameters are chosen by solving an optimization problem. In this observation scheme, the reset conditions are  zero crossing  and sector condition. Furthermore, reset observers have been improved and extended to nonlinear systems, Multi Input Multi Output systems, and time-varying delayed systems \cite{paesa2011reset,paesa2012reset}. Besides, an adaptive reset observer method is proposed in \cite{zhou2014adaptive} and a reset observer based on the delay-dependent approach is developed in \cite{meng2015fast}. A single Lyapunov function is considered for stability analysis because the closed-loop error dynamics is hybrid, the Lyapunov function in both flow set and jump set should be decreasing. Note that this approach  is very conservative. For obtaining less conservative results, piecewise Lyapunov functions are used, in this approach, each Lyapunov function is only needed to be decreasing in a region of the state space. As a result, less conservative stability analysis results are obtained. In \cite{meng2015fast} a Proportional-Integral observer is used and reset strategy is applied to the observer for fault detection purposes. The conventional reset law (zero crossing) and after reset value (jump to zero) is used and the results demonstrate that the fault estimation and the residual convergence to zero can be strikingly improved. 
		
		In this paper, reset strategy is extended to the UIO and a novel sort of UIOs called Reset UIO is proposed. A suitable after-reset value along with a proper jump sector is obtained using LMI approach. Furthermore, the stability analysis for the reset error dynamics is given. A R-UIO is designed in two steps wherein the first step, for our main purpose we characterize the case assuming that all the system states are available. Then the reset law is designed by LMI and the parameters are obtained. In the second case, it is assumed that only the outputs are available but the bounds on the estimation errors are known and the R-UIO is designed. It has been shown that exploiting the reset mechanism in the UIO can improve the performance of the observer. 
		
		The remainder of the paper is organized as follows: in Section \ref{sec:problem}, conventional approach of designing the base UIO is investigated. In Section \ref{sec:R-UIO}, the reset UIO in the ideal case is designed first and then the non-ideal case is presented and stability analysis is provided. In Section \ref{sec:simulation}, simulation results are displayed to validate the performance of the proposed observation strategy. Finally, the concluding remarks are provided in Section \ref{sec:conclusion}.
		%############################################################		

	%############################################################
	\section{Conventional UIO (C-UIO)} \label{sec:problem}
		Consider the system:
		
		\begin{align} \label{general-sys}
			\left\{
				\begin{array}{ccl}
					\dot{x}  & = & Ax+Bu+Dv \\
					y & = & Cx
				\end{array}
			\right.
		\end{align}
		
	\noindent	where $x \in {\rm I\!R}^n$, $u \in  {\rm I\!R}^q$, $v \in {\rm I\!R}^m$ and $y \in {\rm I\!R}^p$ are the state vector, known input vector, unknown input vector and output of the system respectively. $A,B,C$ and $D$ are known matrices with appropriate dimension. Without loss of 		generality, it is assumed that $D$ is of full column rank \cite{chen2012robust}. 
%	\subsection{Conventional Unknown Input Observer (C-UIO)}

	For the state estimation of the aforementioned system a full-order C-UIO can be defined as:
		\begin{align} \label{C-Obs}
%		\nonumber
			\left\{
				\begin{array}{ccl}
					\dot{Z}     & = & NZ+Gu+Ly \\
					\hat{x}     & = & Z-Ey\\
					\end{array}	
					\right.		
		\end{align}
		where $Z\in {\rm I\!R}^n$ is the state of this full-order observer, $\hat{x}\in {\rm I\!R}^n$ is the estimated state vector and $N,G,L,E$ are design matrices for unknown input decoupling goal and other required performances. The parameters of the C-UIO observer can be obtained using \cite{chen2006unknown}:
			
		\begin{align}
		 \label{UIO-Cond}
			\left\{
			\begin{array}{ccl}
						N  & = & MA-KC\\
						G  & = & MB\\
						L  & = & K(I+CE)-MAE\\
						M  & = & I+EC \\
						MD & = & 0\\
			\end{array}
			\right.
		\end{align}
			It is assumed that  $rank(CD)=rank(D)$ and the pair $ (C,MA)$ is detectable. Using the last equation in (\ref{UIO-Cond}), $E$ can be obtained as $$	E   = -D(CD)^+ +Y(I-(CD)(CD)^+)$$
			 in which, $(CD)^+$ is defined as $(CD)^+=((CD)^T(CD))^{-1}(CD)^T$ and Y can be arbitrarily chosen,	and $K$ is a chosen such that $N$ is Hurwitz \cite{chen2012robust}.

	Defining the estimation error as $$e=\hat{x}-x$$
		\noindent the continuous error dynamics can be obtained as:
		\begin{align*}
			\dot{e}& =\dot{\hat{x}}-\dot{x}=Ne
		%   & =\dot{Z}-(I+EC)\dot{x}=\dot{Z}-M\dot{x}\\
		%   & =NZ+Gu+LCx-MAx-MBu-MDv
		\end{align*}	
%		using (\ref{UIO-Cond}) results in:
%		\begin{align*}
%		\dot{e}&=NZ+(LC-MA)x\\
%		&=N(\hat{x}+ECx)+\left[ K(I+CE)-MAE\right]Cx-MAx\\
%		&=N\hat{x}+MAECx-KCECx+KCx+KCECx-MAECx-MAx\\
%		&=N\hat{x}+(KC-MA)x\\
%		&=N\hat{x}-Nx=Ne
%		\end{align*}
	 The above error dynamics indicates that the estimation error converges asymptotically to zero and thus $\hat{x}\longrightarrow x$.
	 
	In the next section, the reset theory is used to introduce a nonlinear observer which can reduce the $L_2$ norm and settling time of the estimation error.

	%############################################################		
	\section{Reset UIO (R-UIO)} \label{sec:R-UIO}
	In this section, R-UIO which is a novel kind of UIO is proposed to estimate the states more rapidly and accurately. The design steps are divided into two cases. In the first case which is called R-UIO with full-state measurement (or ideal case), it is assumed that all the system states can be measured. Then this case is extended to the second approach named R-UIO with partial state measurement (or non-ideal case) in which only the outputs are available.

			\subsection{R-UIO with full-state measurement}
	In this part reset action is added to the C-UIO to improve the performance of the observer. Thus, the R-UIO can be formulated as:		

		\begin{align} \label{R-Obs}
					\nonumber
					\left\{
					\begin{array}{ccl}
					\dot{Z}     & = & NZ+Gu+Ly \\
					\hat{x}     & = & Z-Ey\\
					\hat{y}     & = & C\hat{x}
					\end{array}
					\hspace{110pt}
					\right\} 
					\text{if} \hspace{5pt}
					e\in \mathcal{F} \hspace{5pt}		   	
					\\
					\left\{
					\begin{array}{ccl}
					Z^+     & = & (M-A_REC)Z-(I-A_R)MEy\\
					\hat{x}^+     & = & 	Z^+-Ey\\
					\end{array}
					\hspace{4pt}
					\right\} 
					\text{if}	 \hspace{5pt}
				e\in \mathcal{J}  \hspace{5pt}
	\end{align}
in which $A_R$ is the after reset matrix, $\mathcal{F} =\{e\in {\rm I\!R}^n|e^TFe \geq 0\}$ is the flow set and $\mathcal{J} =\{e\in {\rm I\!R}^n|e^TFe \leq 0\}$ is the jump set and as soon as 	$e\in \mathcal{J}$ jump will happen. It's worth noting that $F$ and $A_R$  will be obtained by solving some inequalities. Note that in the defined R-UIO, all the system states are required to define the jump set.

	For the discrete error dynamics one has:
	\begin{align}
		e^+&=\hat{x}^+-x \nonumber\\
		&=Z^+-Ey-x=Z^+-(I+EC)x
	\end{align}
substituting $Z^+$ from \eqref{R-Obs} results in
	\begin{align}
		  e^+ &=(M-A_REC)Z-(I-A_R)MECx-(I+EC)x 
	 \end{align}
	 using $Z=\hat{x}+ECx$ implies that
    \begin{align}
    	e^+=M(\hat{x}+ECx)-A_REC(\hat{x}+ECx)-MECx+A_RMECx-Mx
    \end{align}
    simplifying the equation leads to
    \begin{align}
% 		   &=M(\hat{x}+ECx)-A_REC(\hat{x}+ECx)-MECx+A_RMECx-Mx\\
%%		   &=M\hat{x}+MECx-A_REC\hat{x}-A_RECECx-MECx+A_RMECx-Mx\\
%		   &=M(\hat{x}-x)-A_REC\hat{x}-A_RECECx+A_R(I+EC)ECx\\
		 e^+  &=Me-A_REC\hat{x}+A_RECx	  
	\end{align}
	adding and subtracting $A_Re$, $e^+$ can be obtained as:
	\begin{align}
%	 e^+&=Me-A_RECe+A_Re-A_Re \nonumber\\
	e^+	&=Me-A_R(I+EC)e+A_Re  \nonumber\\
		&=(A_R-A_RM+M)e.
	\end{align}
	 Therefore, defining $H=A_R-A_RM+M$, the error dynamics can be written as:
		\begin{align} \label{error-dynamic}
			\left\{
			\begin{array}{ccccc}
				\dot{e} & = & Ne & \text{if} & e\in \mathcal{F}  \\
				 e^+ & = & He & \text{if} & e\in \mathcal{J}  
			\end{array}
			\right.
		\end{align}

	Based on reset error dynamics the following theorem on the convergence of R-UIO can be stated:
	\begin{theorem} \label{trm:stab}
		For the system (\ref{general-sys}), if there exists symmetric matrices $P>0$,$F$ and matrix $Q$ and constants $\lambda_f,\tau_f,\tau_j,\tau_w>0$ and $0<\lambda_j\leq1$ such that 
		\begin{subequations}
			\begin{align} 
			\label{LMI-stab1} 
			N^TP+PN+\lambda_fP +\tau_f F&<0   \vspace{15pt}\\ 
			\begin{bmatrix} \label{LMI-stab2}
			\lambda_jP+\tau_jF & (Q-QM+PM)^T\\
			Q-QM+PM		  & P
			\end{bmatrix}&\geq 0 \vspace{15pt}\\  
			H^TFH+\tau_wF&>0 \label{LMI-stab3}
			\end{align}
		\end{subequations}
		the error dynamics is well-posed and the R-UIO given by (\ref{R-Obs}) makes the error converges to zero asymptotically for any initial condition.
	\end{theorem}
	\begin{proof}

   To prove the quadratic stability, consider the following Lyapunov function:
   \begin{align}
   \begin{array}{ccl}
   V(e)      & =e^TPe\\
   \end{array}
   \end{align}
   where $P=P^T>0$. The error dynamics (\ref{error-dynamic}) is asymptotically stable if:
   
   \begin{align} \label{Stab-cond}
   \left\{
   \begin{array}{cclrc}
   \dot{V}(e) & <  & -\lambda_fV(e) & \text{if}& e^TFe\geq 0\\  
   V(	e^+) & \leq & \lambda_jV(e) & \text{if}& e^TFe\leq 0 .\\  
   \end{array}
   \right.
   \end{align}
   The continuous error dynamics inequality in (\ref{Stab-cond})  can be rewritten as:
   \begin{align}
   \label{LMI-1}
   \dot{V} & =\dot{e}^TPe+e^TP\dot{e} <-\lambda_fe^TPe \nonumber\\
   & = (Ne)^TPe+e^TP(Ne) < -\lambda_fe^TPe \nonumber \\
   & = e^T(N^TP+PN+\lambda_fP)e<0 \nonumber \\
%   & N^TP+PN+\lambda_fP <0
   \end{align}
   if $e^TFe\geq0$ holds.
   Using the S-procedure and taking $ e^TFe\geq 0$ into account results in  
   \begin{align}
   N^TP+PN+\lambda_fP +\tau_f F<0, \quad \tau_F\geq0
   \end{align}
   
   Similarly for the discrete error dynamics it is stated that
   \begin{align}
   \label{LMI-2}
   V(e^+)-\lambda_jV(e) \leq  0  \nonumber\\
   (He)^TP(He)-\lambda_je^TPe\leq 0 \nonumber\\
   e^T(H^TPH-\lambda_jP)e\leq 0 \nonumber\\
    H^TPH-\lambda_jP\leq 0 
    \end{align}
    when $e^TFe\leq0$ is satisfied, and with the aid of S-procedure the condition  $ e^TFe\leq0$ can be added to (\ref{LMI-2}) as:
    \begin{align}  \label{LMI-2-S-pro}
    H^TPH-\lambda_jP-\tau_j F\leq 0 
     \end{align}  
   Using the Schur complement lemma the inequality (\ref{LMI-2-S-pro}) can be rewritten as:
   \begin{align}
   \label{LMI-2-rw}
   \begin{bmatrix}
   \lambda_jP+\tau_jF & H^T\\
   H		  & P^{-1}
   \end{bmatrix}\geq 0
   \end{align} 	
    pre and post multiplying (\ref{LMI-2-rw}) by 
   \[
   \begin{bmatrix}
   I & 0\\
   0 & P
   \end{bmatrix}
   \] 
   results in
   \begin{align} \label{eq:LMI-2-rw2}
   \begin{bmatrix} 
   \lambda_jP+\tau_jF & H^TP\\
   PH		  & P
   \end{bmatrix}\geq 0.
  \end{align}
  Replacing $H$ in the (\ref{eq:LMI-2-rw2})	results in:
    \begin{align} \label{eq:LMI-2-rw3}
    \begin{bmatrix} 
    \lambda_jP+\tau_jF & A_R^TP-M^TA_R^TP+M^TP\\
    PA_R-PA_RM+PM		  & P
    \end{bmatrix}\geq 0
    \end{align}
The inequality (\ref{eq:LMI-2-rw3}) is not linear since it contains multiplication of unknown parameters $P$ and $A_R$. Therefore, using the variable change $Q=PA_R$, one gets
   
   \begin{align} \label{LMI-2-rw2}
   \begin{bmatrix}
   \lambda_jP+\tau_jF & (Q-QM+PM)^T\\
   Q-QM+PM		  & P
   \end{bmatrix}\geq 0
   \end{align}
Moreover, for the well-posedness of the system it is required that after a jump, the error trajectory jumps out of the jump set i.e:
 \begin{align} \label{LMI-3}
 \begin{array}{ccl}
	 (e^+)^TF(e^+)>0 & \text{if} & e^TFe\leq 0 
	 \end{array}
 \end{align}
 thus, using S-procedure
  \begin{align} \label{LMI-3-2}
  \begin{array}{ccl}
  H^TFH+\tau_wF>0 
  \end{array}
  \end{align}
must holds and this completes the proof.
\end{proof}
\begin{remark}
	It's worth mentioning that the inequality (\ref{LMI-3-2}) is checked a posteriori, in practice. It means that as $H$ and $F$ are obtained previously in (\ref{LMI-stab1}) and (\ref{LMI-stab2}), if there is $\tau_w$ such that the inequality (\ref{LMI-3-2}) holds then the system is well-posed and in this case, the reset will be applied to the system.
\end{remark}

As it has been mentioned before, the ideal case is considered to design the matrices $F$, $P$ and $A_R$. It means that if all the states are available the mentioned matrices can be obtained by solving the LMIs \eqref{LMI-stab1} and \eqref{LMI-stab2}. But the problem with the designed R-UIO in (\ref{R-Obs}) is that the flow and jump sets depend on the estimation error $e$ which is not available in general. Moreover, in this observer, the inequality (\ref{LMI-3-2}) should be checked a posteriori and it may not be satisfied in some cases.
%\begin{assumption} \label{ass:error-bounds}
%	It is assumed that the bounds of each state which is not measured is known. 
%\end{assumption}

% Considering the Assumption \ref{ass:convex-comb}, although some of the states are not available, the bounds of them are known. Thus, this bounds can be used to design the reset law. Furthermore, the condition (\ref{LMI-3-2}) is modified such that there is no need to check any condition posteriorly.

\subsection{R-UIO with partial state measurement}
So far it has been assumed that all the states can be measured to design the reset law parameters. Although, the estimation errors are available, this is not the case since only some of them can be measured in practice and an observer is designed to estimate the unmeasured states. The problem in the ideal case formulation is that the error is used to decide whether jump happens or not ($e^TFe\leq 0$), but it is not available in general. To cope with this problem assume that error bounds are available instead of the exact error and use these bounds to decide about jump instants.
\begin{assumption} \label{ass:convex-comb}
	Suppose that a polytope $\mathcal{S}\subset {\rm I\!R}^n$ is known such that $ e(t_0)\in \mathcal{S}$. Denote with $e_{v_i}$ its vertices and $i=1,..,N_v$ where $N_v$ is the number of vertices.
\end{assumption}
\begin{remark}
	Assumption 1 could be relaxed to just suppose to know a bound on $e(0)$. In fact, if a general non-polytopic boundary set is known (for instance a bound on the norm) then it is possible to find a polytope including the boundary set.
\end{remark}
Note that Assumption \ref{ass:convex-comb} implies that $e(0)$ is the convex combination of $e_{v_i}$, i.e. $\exists \alpha_{v_i}\geq 0$ such that $e(0)=\sum_{i}^{}\alpha_{v_i} e_{v_i}$ and $\sum_{i}^{}\alpha_{v_i}=1$.
 ًRemember that from Assumption \ref{ass:convex-comb} it is supposed that the vertices of the bounding set containing the $e(0)$ are known. Given a vertex as the initial condition and the set of reset instants then this provides a trajectory which we call $e_{v_i}(t)$. Our objective is to give a criterion such that

  1. $\quad e(t)\subseteq \textbf{conv} \{e_{v_i}(t)\}$ for $t \in {\rm I\!R}^+$ and

  2. All the trajectories $e_{v_i}(t)$ are bounded and converge to 0.
  
 Thus, the convergence of them results in the convergence of the real error. Therefore, it is necessary to design an appropriate reset law such that the stability of observer is guaranteed.
 
For example, suppose that only one of the states is not available but the error bounds are known, let's say the estimation error of this state is in the interval $[e_{v_1}, e_{v_2}]$. Starting from the vertices $e_{v_1}$ and $e_{v_2}$ results in two error trajectories $e_{v_1}(t)$ and $e_{v_2}(t)$ (dash-lines in Figure \ref{bound-traj}). Suppose that jump happens at $t_2$, at this moment although the trajectory $e_{v_2}(t)$ is inside the jump sector, the trajectory $e_{v_1}(t)$ is outside and in this case stability is not guaranteed. 

On the other hand, suppose that the jump set is $\mathcal{J}=\{e\in {\rm I\!R}^n|e^TFe \leq 0\}$ and reset will happen when $e_{v_1},e_{v_2}\in \mathcal{J}$. This means that both error trajectories should be inside the jump sector and jump happens at $t_1$. Even in this case, the convergence of $e(t)$ might not be satisfied. Since the behavior of the $e_{v_2}(t)$ while flowing in the jump set in the interval $t_1-t_2$ is not known and if the reset happens in a wrong moment it may destabilize the system.

%Consider the Figure \ref{bound-traj}, jump sector is inside the solid red lines. Suppose that one of the trajectories enters the jump sector at $t_{e_{v_2}}$ but the other doesn't. In this case, reset will not happen until the other trajectory enters the jump sector ($t_{e_{v_1}}$) and in this period (T or waiting period), the behavior of the system is unknown and if the reset happens in a wrong moment it may destabilize the system. 

An example is given to demonstrate the importance of choosing the right jump moment even if there is only one single vertex and all the states are measured.

%   Although this reset law is conservative, stability is guaranteed. Since both error trajectories are inside the jump sector, according to Theorem \ref{trm:stab}, the error dynamic is stable.
  \begin{figure} [!ht]
  	\centering
  	\includegraphics[width=8cm,height=6cm]{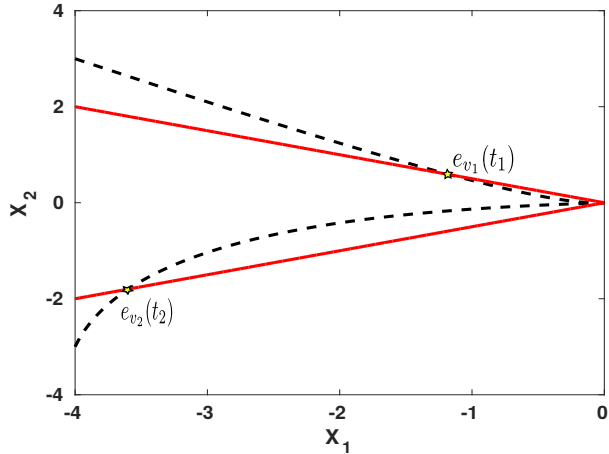}
  	\caption{boundary trajectories and jump sector}
  	\label{bound-traj}
  \end{figure}
  \textbf{Illustrative example}: Consider the second order reset system given by
  		\begin{align} 
  		\left\{
  		\begin{array}{ccccc}
  		\dot{e} & = & Ne & \text{if} & e^TFe\leq 0\\
  		e^+ & = & He  &  \text{if} & e^TFe \geq 0 \\
  		e(0)&=&[-15 ,10]^T
  		\end{array}
  		\right.
  		\end{align}
  		in which
  			\[
  			N=\begin{bmatrix}
  			-0.1 & 1 \\
  			-1 & -0.1\\
  			\end{bmatrix}, 
  			H=\begin{bmatrix}
  			0 & 0.4\\
  			-2 & 0 \\
  			\end{bmatrix},
  			F=\begin{bmatrix}
  			0 & 1 \\
  			1 & 0
  			\end{bmatrix}
  			\]
  			In this example, the error trajectory starts from the initial condition and the wrong jump happens when the trajectory is inside the jump sector and $|e_2|=0$. In Figure \ref{fig:unstable-sys} the phase portrait and the Lyapunov function $V=e^TPe$ with $	P=\begin{bmatrix}
  			1.3296 & 0 \\
  			0 & 0.2924
  			\end{bmatrix}$  of the system demonstrate that a wrong jump instant in the jump set can destabilize the system. Therefore, it's very important to choose the jump moment carefully. It is worth mentioning that in this example there is just one error trajectory since the initial value of both states are exactly known.
  			  \begin{figure} [!ht] 
  			  	\begin{subfigure}[t]{0.45\textwidth}
  			  		\centering
  			  		\setlength{\figW}{4.5cm}
  			  		\setlength{\figH}{4cm}			
\definecolor{mycolor1}{rgb}{0.00000,0.44700,0.74100}%
\begin{tikzpicture}

\begin{axis}[%
width=0.951\figW,
height=\figH,
at={(0\figW,0\figH)},
scale only axis,
xmin=-80,
xmax=60,
xlabel style={font=\color{white!15!black}},
xlabel={$\text{e}_\text{1}$},
ymin=-100,
ymax=150,
ylabel style={font=\color{white!15!black}},
ylabel={$\text{e}_\text{2}$},
axis background/.style={fill=white},
legend style={legend cell align=left, align=left, draw=white!15!black}
]
\addplot [color=mycolor1, line width=2.0pt]
  table[row sep=crcr]{%
-15	10\\
-14.1685149901306	10.952510787305\\
-13.2783265617932	11.8385573947375\\
-12.3345893802521	12.6546599341158\\
-11.3426533801663	13.3977285609366\\
-10.3080348687942	14.0650735176214\\
-9.2363871927711	14.6544130760254\\
-8.13347111851496	15.1638793679088\\
-7.00512507574177	15.5920221026066\\
-5.85723541226912	15.9378101815346\\
-4.69570680627271	16.200631229391\\
-3.52643297945369	16.3802890718926\\
-2.35526785119879	16.4769991995979\\
-1.18799726979464	16.4913822667449\\
-0.0303114521200598	16.4244556830575\\
1.11222174198113	16.2776233650847\\
2.23418258021849	16.0526637218109\\
3.33032415912587	15.751715956974\\
4.39559592871444	15.3772647777174\\
5.4251659403445	14.9321236058574\\
6.41444172525635	14.4194163941396\\
7.35908971916038	13.8425581553664\\
8.25505315655587	13.2052343171848\\
9.09856836698461	12.5113790196074\\
9.88617941418181	11.7651524759933\\
10.6147510280189	10.9709175212276\\
11.2814797881919	10.133215473199\\
11.8839035277513	9.25674143538805\\
12.4199089337482	8.34631916943681\\
12.8877373314389	7.40687566698307\\
13.2859886476055	6.44341554981335\\
13.6136235575638	5.46099542652655\\
13.8699638293058	4.46469833242051\\
14.0546908869129	3.45960837722922\\
14.1678426238429	2.45078572266877\\
14.2098085048983	1.44324200851568\\
14.1965943877999	0.72709046091191\\
0.233758970629196	-28.3793373960424\\
-1.18065682290577	-28.2210702857639\\
-2.57872439130716	-27.993326500008\\
-3.95708234489774	-27.6973359764224\\
-5.31245065750505	-27.3344886537305\\
-6.64163814584566	-26.9063294626967\\
-7.94154965472961	-26.4145529746782\\
-9.20919293279653	-25.8609977244633\\
-10.4416851844197	-25.2476402247528\\
-11.6362592843628	-24.5765886902555\\
-12.7902696427509	-23.8500764899273\\
-13.9011977089126	-23.070455346393\\
-14.966657103668	-22.2401883020443\\
-15.9843983706699	-21.361842471712\\
-16.9523133384515	-20.4380816021564\\
-17.8684390858942	-19.4716584589162\\
-18.7309615048922	-18.4654070612947\\
-19.5382184550656	-17.4222347864495\\
-20.2887025064464	-16.3451143636797\\
-21.1124529616703	-15.0120280716329\\
-21.850483217017	-13.6397498002406\\
-22.5010953603214	-12.2336884996838\\
-23.0629269876151	-10.7993139997767\\
-23.5349523272374	-9.34213653130465\\
-23.9164821333601	-7.86768633343169\\
-24.2071623632637	-6.38149341997675\\
-24.4069716569462	-4.88906757595601\\
-24.5162176417599	-3.39587865414777\\
-24.5355320887732	-1.9073372395583\\
-24.4989026457104	-0.920195312389524\\
-0.269714436248023	48.9672113923773\\
1.68414689821603	48.7530702146599\\
3.62049542874834	48.4620001266132\\
5.53634396129584	48.0950527896616\\
7.42876095755548	47.6533941528847\\
9.29487483670752	47.1383017087827\\
11.1318781479624	46.5511615901073\\
12.9370316081768	45.8934655136433\\
14.7076679990557	45.1668075770208\\
16.4411959187245	44.3728809148341\\
18.1351033827302	43.5134742205206\\
19.7869612698094	42.5904681406197\\
21.3944266080486	41.605831548188\\
22.9552456973489	40.5616177022909\\
24.467257064406	39.4599603006214\\
25.9283942467125	38.3030694324145\\
27.3366884023909	37.0932274389345\\
28.6902707429733	35.8327846889038\\
29.9873747865479	34.5241552763247\\
31.226338429003	33.1698126482133\\
32.4056058314099	31.7722851698209\\
33.5237291218943	30.3341516349628\\
34.5793699106622	28.8580367291024\\
35.5713006171516	27.346606452859\\
36.4984056085972	25.8025635136133\\
37.3596821496006	24.2286426928769\\
38.1542411626073	22.6276061970723\\
38.8813077994992	21.0022389993422\\
39.5402218248109	19.3553441799594\\
40.1304378113792	17.6897382728558\\
40.6515251495319	16.0082466257229\\
41.1031678712136	14.313698781055\\
41.4851642907324	12.6089238854205\\
41.7974264640965	10.8967461341421\\
42.0399794691858	9.17998025845924\\
42.2129605092743	7.46142706212015\\
42.3166178426872	5.74386901422199\\
42.3513095416324	4.03006590497318\\
42.3175020835027	2.32275057090072\\
42.2475348310912	1.04814819236785\\
0.249849878345827	-84.4315375563677\\
-2.27877213784419	-84.1606722850773\\
-4.79094859982169	-83.8149954851964\\
-7.28449393906286	-83.3952461519068\\
-9.75725191700877	-82.9022262779461\\
-12.2070974144589	-82.3367997678125\\
-14.6319381825581	-81.6998913019768\\
-17.0297165540064	-80.9924851524159\\
-19.3984111131658	-80.2156239508195\\
-21.7360383237829	-79.3704074108558\\
-24.0406541130907	-78.4579910059172\\
-26.3103554111013	-77.4795846037985\\
-28.5432816439478	-76.436451059792\\
-30.7376161801816	-75.3299047697101\\
-32.8915877289803	-74.161310184379\\
-35.003471689272	-72.9320802871673\\
-37.0715914488304	-71.6436750361432\\
-39.0943196324468	-70.2975997724736\\
-41.0700792983363	-68.8954035967007\\
-42.9973450819871	-67.4386777145527\\
-44.8746442867143	-65.9290537539625\\
-46.7005579202311	-64.3682020549835\\
-48.4737216766051	-62.7578299343101\\
-50.1928268630184	-61.0996799261203\\
-51.8566212708047	-59.3955280009715\\
-53.4639099902916	-57.6471817644899\\
-55.0135561690252	-55.8564786376031\\
-56.5044817130148	-54.0252840200709\\
-57.9356679306824	-52.1554894390762\\
-59.3061561192609	-50.2490106846402\\
-60.6150480934341	-48.3077859336279\\
-61.8615066560699	-46.3337738641119\\
-63.0447560109462	-44.3289517618591\\
-64.164082117427	-42.2953136207024\\
-65.218832987096	-40.2348682385576\\
-66.2084189224088	-38.149637310835\\
-67.1323126974771	-36.0416535229922\\
-67.9900496811495	-33.9129586439622\\
-68.7812279026059	-31.7656016221822\\
-69.5055080597328	-29.6016366859358\\
-70.1626134705961	-27.4231214497061\\
-70.7523299683813	-25.2321150282251\\
-71.2745057402157	-23.0306761598862\\
-71.7290511103386	-20.8208613411677\\
-72.1159382681331	-18.6047229736995\\
-72.4352009415799	-16.3843075255796\\
-72.6869340167381	-14.1616537085297\\
-72.8712931039075	-11.938790672452\\
-72.9884940511668	-9.71773621892663\\
-73.0388124060304	-7.50049503516325\\
-73.0225828260052	-5.28905694989183\\
-72.9401984388747	-3.08539521265081\\
-72.7921101535744	-0.891464797900454\\
-0.0650608926267182	145.4564653828\\
};
%\addlegendentry{data1}

\addplot [color=red, draw=none, mark=asterisk, mark options={solid, red}]
  table[row sep=crcr]{%
-15	10\\
};
%\addlegendentry{data2}

\addplot [color=red, draw=none, mark=triangle, mark options={solid, rotate=270, red}]
  table[row sep=crcr]{%
-0.0650608926267182	145.4564653828\\
};
%\addlegendentry{data3}

\addplot [color=red]
  table[row sep=crcr]{%
-80	0\\
60	0\\
};
%\addlegendentry{data4}

\addplot [color=red]
  table[row sep=crcr]{%
0	-100\\
0	150\\
};
%\addlegendentry{data5}

\end{axis}

\begin{axis}[%
width=1.227\figW,
height=1.227\figH,
at={(-0.16\figW,-0.135\figH)},
scale only axis,
xmin=0,
xmax=1,
ymin=0,
ymax=1,
axis line style={draw=none},
ticks=none,
axis x line*=bottom,
axis y line*=left,
legend style={legend cell align=left, align=left, draw=white!15!black}
]
\node[below right, align=left, font=\bfseries]
at (rel axis cs:0.705,0.749) {$J$};
\node[below right, align=left, font=\bfseries]
at (rel axis cs:0.712,0.269) {$F$};
\node[below right, align=left, font=\bfseries]
at (rel axis cs:0.167,0.26) {$J$};
\node[below right, align=left, font=\bfseries]
at (rel axis cs:0.168,0.751) {$C$};
\end{axis}
\end{tikzpicture}%
%  			  	\centering
%  			  	\includegraphics[width=5cm,height=4.5cm]{unstable_example}
				\caption{Phase portrait}
  			  	\label{fig:unstable_example}
  			  	\;
  			  	\end{subfigure}
  			  	 	\begin{subfigure}[t]{0.4\textwidth}
  			  	 		\centering
  			  	 		\setlength{\figW}{4.5cm}
  			  	 		\setlength{\figH}{4cm}			
  			  	 		\input{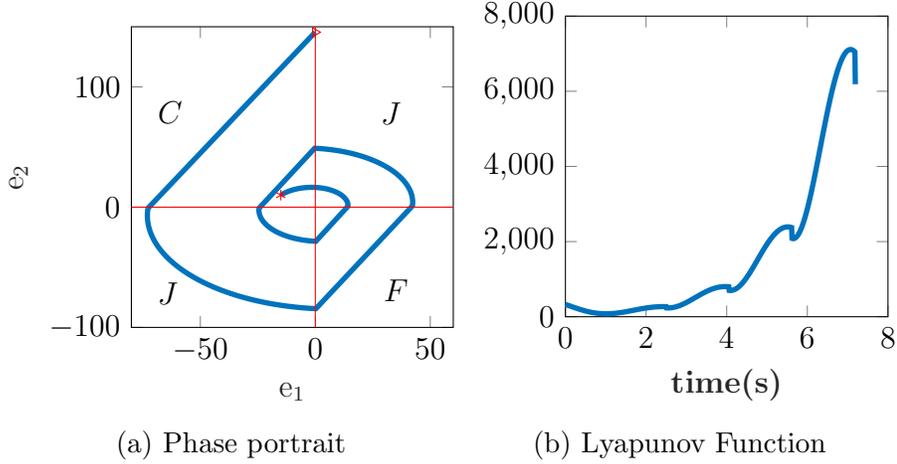}
  			  	 		\caption{Lyapunov Function}
  			  	 		\label{fig:unstable_example_lyap}
  			  	 	\end{subfigure}
  			  	 	\caption{Destabilizing the system with wrong jump moment}
  			  	 	\label{fig:unstable-sys}
  			  \end{figure}

  	An additional constraint is needed to overcome the aforementioned problem to guarantee the asymptotic stability of the error dynamics. Next Theorem addresses this 
  	issue.
 
	\begin{theorem} \label{trm:modified-stab}
	Consider the reset system
			\begin{align} \label{eq:err_dyn_time}
			\left\{
			\begin{array}{ccccc}
			\dot{e}(t) & = & Ne(t) & \text{if} & t\notin T_R\\\
			e(t^+) & = & He(t)  &  \text{if} & t\in T_R
			\end{array}
			\right.
			\end{align}
	in which
	\begin{align}
		T_R\in\{\{t_k\}_{k=0}^\mathcal{N}: t_k>t_{k-1}, \mathcal{N}\in \mathbb{N}\cup \{\infty\}\}
	\end{align}  is the reset times sequence. If the function $V(e)=e^TPe$ satisfies the inequalities (\ref{LMI-stab1}) and (\ref{LMI-stab2}) of Theorem \ref{trm:stab} and $T_R$ is such that
	\begin{align} \label{eq:lyap-trm2}
		V(e(t_k^-))\leq (1-\epsilon)V(e(\tau_k)) \quad \forall t_k\in T_R
	\end{align} 
	with $ \epsilon\in (0,1)$ and
	\begin{align} \label{TAU_K}
	\tau_k=\min\{t\in{{\rm I\!R}}^+|e(t)^TFe(t)\leq 0, \;  \; t\geq t_{k-1}\}
	\end{align}
	holds
	 and $e(t_k^-)Fe(t_k^-)\leq 0$ for all $t_k \in T_R$ 
	  then  the system (\ref{eq:err_dyn_time}) is asymptotically stable.
  
\end{theorem}

\begin{proof} 	
	Note that besides the function $V(e)=e^TPe$  obtained from the Theorem \ref{trm:stab}, a Lyapunov function $V_n(e)=e^TP_ne$ exists for the nominal system such that $\dot{V}_n\leq -\lambda_n V_n$ for $\lambda_n\geq 0$ since the system is detectable.
	
	To assert asymptotic stability, we first prove the inequalities 
 		\begin{align} 
 		V(e(t)) & \leq  \beta V(e(t_{k-1}^+)) 
 		\qquad \forall t \in (t_{k-1},t_k)
 		\label{eq: as-stab0}\\
 		V(e(t_k^+))& \leq  \lambda_j(1-\epsilon)V(e(t_{k-1}^+)
 		\label{eq: as-stab1}
 		\end{align}

 	\noindent with $\beta$ defined as:
 	 	\begin{align} 	
 	 		\beta&={\max_{e\in \varepsilon_n(\gamma)} \; V(e)}
 	  	\end{align}
 	  		 in which
 		\begin{align}
	  		\gamma&={\max_{e\in \varepsilon(1)} \; V_n(e)}\\
	  		\varepsilon(\alpha)&=\{e\in {\rm I\!R}^n|   V(e) \leq \alpha\}\\ 
 	  		\varepsilon_n(\alpha)&=\{e\in {\rm I\!R}^n| V_n(e) \leq \alpha\}.
  		\end{align}

 	 To prove (\ref{eq: as-stab0}), note that from the definition of $\beta$ and $\gamma$ it can be inferred that $\varepsilon(\beta)\geq
	\varepsilon_n(\gamma)\geq\varepsilon(1)$. An illustrative example of the sets is given in Figure \ref{fig:levelsets}. Since $\dot{V}_n\leq -\lambda_n V_n$, 
	any trajectory starting in $\varepsilon_n(\gamma)$ at $t=t_0$ stays in $\varepsilon_n(\gamma)$ for all $t\geq t_0$ while flowing.
	 Therefore, it remains in $\epsilon(\beta)$ too and between two jumps can't leave it. As a result, the function $V$ may increase  $\beta$ times during the flow but not more. 
	 
	 After a jump, because of the (\ref{LMI-stab2}),  $V$ is decreasing again, and consequently, the error cannot go further than $\varepsilon(\beta)$ and remains bounded when starting in $\varepsilon(1)$. Due to the homogeneity, this reasoning can be extended to other level sets leading to \eqref{eq: as-stab0} when flowing.
	\begin{figure} 
					\centering
		\includegraphics[width=4.5cm,height=4.2cm]{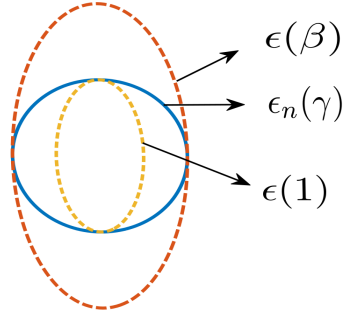}
		\caption{level sets of $V$ and $V_n$}
		\label{fig:levelsets}
	\end{figure}
	
	To prove \eqref{eq: as-stab1}, two possibilities should be considered:
%	 both finite and infinite number of jumps are considered. 	
%		\textbf{Finite $\mathcal{N}$:}
%		It can happen that at some point there are no more jumps. In this case, since the error dynamics is stable, in flow set it acts like the nominal system and there exists a function  such that: $$\dot{V}\leq-\lambda_fV$$	
%	or we can say since $\dot{e}_i=Ne_i$ and $N$ is Hurwitz, $e_i(t)=e^{Nt}e_i(0)$ converges to zero asymptotically. According to Lemma \ref{lem:convex-comb},
%	$$	e(t)   =\alpha_1 (e^{Nt}e_1(0))+\alpha_2 (e^{Nt}e_2(0))+...\alpha_i (e^{Nt}e_i(0))$$
%	Therefore, $e(t)$ converges to zero also.	
%	\textbf{Infinite $\mathcal{N}$:}	
%	In this case, there is always a jump. Thus, it should be shown that the function $V$ is decreasing at jump instants, i.e:
%	\begin{align}\ \label{eq:as-stab-jump}
%			V(e(t_k^+)) \leq & \lambda_j(1-\epsilon)V(e(t_{k-1}^+))
%	\end{align}
	
%	To prove the asymptotic stability,  stability and attractivity should be proved. For stability, it should be showed that the error dynamics is bounded in the waiting period. For attractivity, decrease at each jump instant should be guaranteed.
	
%	For proving (\ref{eq:as-stab-jump}) two possibilities should be considered:
	
\noindent	1. After a jump, the error trajectory is in the flow set. In this case one has
	\begin{align}
		\dot{V}\leq-\lambda_f V \qquad t \in (t_{k-1}^+, \tau_k)
	\end{align}
	from the definition of $\tau_k$ in (\ref{TAU_K}), resulting in
	\begin{align} \label{eq:sec_jump}
		V(e(\tau_k))\leq e^{-\lambda_f(\tau_k-t_{k-1})}V(e(t_{k-1}^+)).
	\end{align}
	Since by construction one has (\ref{eq:lyap-trm2}) and
	\begin{align}
		V(e(t_k^+)) & \leq \lambda_jV(e(t_k^-)) \qquad t \in (\tau_k,t_k^-)
		\label{eq:jump_decrease} 
	\end{align}
%	which (\ref{eq:lyap-trm2}) means that at every jump, the function $V(e(t_k^-)$ should be less than its value at entrance instant $(V(e(\tau_k))$. This also prevents Zeno solution and there is no need to check any condition posteriorly as in  (\ref{LMI-stab3}). 
	holds from (\ref{LMI-stab2}), hence,
	\begin{align} 
		V(e(t_k^+)) \leq \lambda_j (1-\epsilon)V(e(\tau_k))
		\label{eq:sec_dec2}
	\end{align}
	Since $-\lambda_f$ is negative
		$ e^{-\lambda_f(\tau_k-t_{k-1})} \leq 1$ and therefore,
		(\ref{eq:sec_dec2}) and (\ref{eq:sec_jump}) result in
		\begin{align} \label{eq:jump_dec}
			V(e(t_k^+)) & \leq \lambda_j (1-\epsilon)V(e(\tau_k))  \leq\lambda_j(1-\epsilon)e^{-\lambda_f(\tau_k-t_{k-1})}V(e(t_{k-1}^+)) \nonumber\\
			& \leq\lambda_j(1-\epsilon)V(e(t_{k-1}^+))
		\end{align}
		It should be noted that in this case there is no Zeno solution since after a jump there is always flowing.
		
		2. If the error trajectory jumps in the jump sector.
		In this case, $\tau_k=t_{k-1}^+$ and then $e^{-\lambda_f(\tau_k-t_{k-1})}=1$. Hence from  (\ref{eq:jump_dec})  $$V(e(t_k^+))\leq\lambda_j(1-\epsilon)V(e(t_{k-1}^+))$$
	Principally, in this case there could be Zeno solution. However, since $\epsilon>0$ and from (\ref{eq:lyap-trm2}), $t_k^->\tau_k$. Hence there is always flowing before next jump  and therefore Zeno cannot happen.		
		
From (\ref{eq: as-stab0}) and (\ref{eq: as-stab1}) $V(e(t))\leq\beta V(e(0))$ is true which implies stability. Moreover, (\ref{eq: as-stab1}) implies attractivity and from this asymptotic stability is inferred.

\end{proof} 
	This Theorem shows that at the reset moment $t_k$ the error trajectory should be inside the sector ($e(t_k^-)Fe(t_k^-)\leq 0$) and the value of the function $V$ should be less than its value at the instant $\tau_k$. Furthermore, $T_R$ is a set of strictly increasing instants of all jumps.
\begin{remark}
	The proof of the Theorem \ref{trm:modified-stab} is valid for both finite and infinite $\mathcal{N}$. If $\mathcal{N}$ is finite, from (\ref{eq: as-stab0}) and (\ref{eq: as-stab1}) the system states remains bounded and after the last jump since the nominal error dynamics is asymptotically stable, the error will go to zero asymptotically. Similarly, if $\mathcal{N}$ is infinite, since (\ref{eq: as-stab0}) and (\ref{eq: as-stab1}) hold for every $t_k, k\rightarrow \infty$, the system \eqref{eq:err_dyn_time} is asymptotically stable.
\end{remark}
\begin{remark}
	Inequality (\ref{eq:lyap-trm2}) guarantees that before the next jump there is a positive time interval of flow which means that there is no Zeno solution and system is well-posed. 
\end{remark}

%\begin{remark}
%	The error dynamic in the {\tiny }waiting period is bounded since the system is linear and it can not go to infinity in finite time.
%\end{remark}

Notice that	$\tau_k$ is the first instant that the error trajectory enters the jump sector after the $k-1^{th}$ jump. If it leaves the sector without any jump the $\tau_k$ is held until a jump happens. After that, a new value for the $\tau_k$ should be considered. Moreover, $t_k$ is obtained when the error trajectory is inside the jump sector and the jump condition is satisfied. It should be noted that if after a  jump, the error trajectory is again inside the jump sector then $\tau_{k+1}$ is equal to $t_{k}$.

%\begin{corollary}

	In the Theorem \ref{trm:modified-stab}, a reset law is proposed to guarantee the boundedness and convergence of a given trajectory. Since many trajectories generated by $e_{v_i}$ as initial condition are considered, it is necessary to impose that not only the reset law of Theorem \ref{trm:modified-stab} is satisfied, but also $e(t)\subseteq \textbf{conv} \{e_{v_i}(t)\}$. The next Theorem addresses this issue.
%\end{corollary}

\begin{theorem} \label{prop:convex-comb}
Suppose Assumption \ref{ass:convex-comb} holds and with $V$ satisfying \eqref{LMI-stab1} and \eqref{LMI-stab2}, the reset system \eqref{eq:err_dyn_time} with reset times sequence $T_R$ is such that
\begin{align}
	\tau_{k_i} =\min\{t\in \mathbb{R}^+|e_{v_i}(t)^TFe_{v_i}(t)\leq 0, \quad t\geq t_{k-1}\}\\
		V(e_{v_i}(t_k^-))\leq (1-\epsilon)V(e_{v_i}(\tau_{k_i})) \quad \forall t_k\in T_R, \epsilon>0
		\label{eq:jump_all}
\end{align}
for all $i=1,...,N_v$, then the reset system \eqref{eq:err_dyn_time} is asymptotically stable.
\end{theorem} 

\begin{proof}
	First note that Assumption \ref{ass:convex-comb} implies that $e(0)$ is a convex combination of $v_i$.
	
	if $t\in (t_{k-1},t_k)$, hence $t\notin T_R$ then
	\begin{align*}
	\dot{e}(t) & =Ne(t)\\
	e(t)   & =e^{Nt}e(t_{k-1}^+)\\
	e(t)   & =e^{Nt}(\alpha_{v_1} e_{v_1}(t_{k-1}^+)+\alpha_{v_2} e_{v_2}(t_{k-1}^+)+...+\alpha_{v_i} e_{v_i}(t_{k-1}^+))\\
	e(t)   &=\alpha_{v_1} (e^{Nt}e_{v_1}(t_{k-1}^+))+\alpha_{v_2} (e^{Nt}e_{v_2}(t_{k-1}^+))+...+\alpha_{v_i} (e^{Nt}e_{v_i}(t_{k-1}^+))
	\end{align*}
	Therefore, $e(t)$ is also a convex combination of the error trajectories. Similarly 	if $t=t_k$, thus $t\in T_R$ then
	\begin{align*}
	e(t_k^+) & =He(t_k^-)\\
	e(t_k^+) & =H(\alpha_{v_1} (e^{Nt}e_{v_1}(t_k^-))+\alpha_{v_2} (e^{Nt}e_{v_2}(t_k^-))+...+\alpha_{v_i} (e^{Nt}e_{v_i}(t_k^-)))\\
	e(t_k^+) & =\alpha_{v_1} (He^{Nt}e_{v_1}(t_k^-))+\alpha_{v_2} (He^{Nt}e_{v_2}(t_k^-))+...+\alpha_{v_i} (He^{Nt}e_{v_i}(t_k^-))
	\end{align*}
Moreover, according to Theorem \ref{trm:modified-stab} since each error trajectory $e_{v_i}(t)$ converges to zero and since $e(t)$ is in the convex hull of the all trajectories $e_{v_i}(t)$  it will converges to zero also and this completes the proof.
\end{proof}
This Theorem states that the jump happens for all the trajectories $e_{v_i}(t)$ at the same instant $t_k$ in which the jump condition \eqref{eq:jump_all} is satisfied. Hence, it can be seen that whether flowing or jumping the real error is a convex combination of the error trajectories and converges to zero.
	%############################################################s
	\section{Simulation} \label{sec:simulation}
	In order to show the effectiveness of the proposed method a numerical example is considered \cite{chen2012robust}. For a fair comparison, C-UIO is designed to be optimal (using LQR method), the goal is to show that the proposed R-UIO outperforms the optimal C-UIO.

	Consider the system \eqref{general-sys} with:
	\[
		A=\begin{bmatrix}
		 -1 & 1 & 0\\
		 -1 & 0 & 0\\
		  0 &-1 &-1
		\end{bmatrix}, 
		B=\begin{bmatrix}
		0 \\
		0 \\
		1
		\end{bmatrix},
		C=\begin{bmatrix}
			1 & 0 & 0\\
			0 & 0 & 1
		\end{bmatrix},
		D=\begin{bmatrix}
		-1 \\
		0 \\
		0
		\end{bmatrix},
	\]
	using the LQR method with weighting matrices equal to identity, the observer gain, $K$ is obtained as:
	
	\[
			K=\begin{bmatrix}
			  1.2926  & 0.3638\\
			  -0.7654 &  -1.0076\\
			  0.3638 &  0.9830
			\end{bmatrix}\]		
choosing 
	\[ 
		Y=\begin{bmatrix}
		1  &       1 \\
		1  &       1 \\
		1  & 	   1 
		\end{bmatrix} 
	\] 
and using \eqref{UIO-Cond}, the observer parameters can be calculated as:
		\[
		E=\begin{bmatrix}
		 -1  &   1\\
		  0  &   1\\
		  0  &   1
		\end{bmatrix},
		N=\begin{bmatrix}
		-1.2926  &  -1.0000 &  -1.3638\\
		-0.2346  &  -1.0000 &  0.0076\\
		-0.3638  & -2.0000  & -2.9830
		
		\end{bmatrix}, 
		G=\begin{bmatrix}
		1 \\
		1 \\
		2
		\end{bmatrix}\]
		\[
		L=\begin{bmatrix}
         0       &4.0202\\
         -1.0000 &  0.2194\\
         0        & 6.3297
		\end{bmatrix}, \\
		M=\begin{bmatrix}
		   0  &   0  &   1\\
		   0  &   1  &   1\\
		   0  &   0  &   2
		
		\end{bmatrix}, 		
		\]
with these parameters, the design of optimal C-UIO is completed.

 Now, to obtain the matrices $P$, $F$ and $A_R$, the ideal R-UIO should be designed by solving the inequalities \eqref{LMI-stab1} and \eqref{LMI-stab2} of Theorem \ref{trm:stab}. It worth noting that, $\lambda_f$, $\lambda_j$, $\tau_f$ and $\tau_j$ are unknown and result in multiplication of parameters. Therefore, to solve these inequalities, a change of variable is used to remove one of them. Consider $\tau_fF=\bar{F}$ thus, $ \tau_jF$ can be replaced with $ \frac{\tau_j}{\tau_f}\bar{F}=\bar{\tau_j}\bar{F}$. It is the same as letting $\tau_f=1$ and solving the inequalities. To deal with the other nonlinearities, a grid is  considered  for $\lambda_f$, $\lambda_j$ and $\tau_j$, then the inequalities are solved at each point of the grid to obtain a feasible solution. Thus, many feasible but sub-optimal solutions may be obtained and a criterion is needed to choose the best one.

Consider $\lambda_{f_{i_f}}=0.1+i_f\Delta_{\lambda_f}$ ,$\Delta_{\lambda_f}=1$, $\lambda_{j_{i_j}}=0.1+i_j\Delta_{\lambda_j},\Delta_{\lambda_j}=0.1$,  $\tau_{j_{i_\tau}}=0.1+i_\tau\Delta_{\tau_j},\Delta_{\tau_j}=1$ and $i_f,i_j,i_\tau<10$ as natural numbers. Now, the inequalities in  Theorem \ref{trm:stab} with the remaining parameters are linear and can be solved using LMI techniques. Solving the LMIs result in 22 feasible solutions which are shown in Table \ref{feas-sol}. As can be seen from Figure \ref{sector} the sector size is directly related to the $\lambda_f$ and $\lambda_j$. This means that, in a fixed $\tau_j$, bigger $\lambda_f$ and $\lambda_j$ result in a bigger sector which in turn increases the jump probability.

\begin{figure} 
%			\centering
			\includegraphics[width=\textwidth,height=8cm]{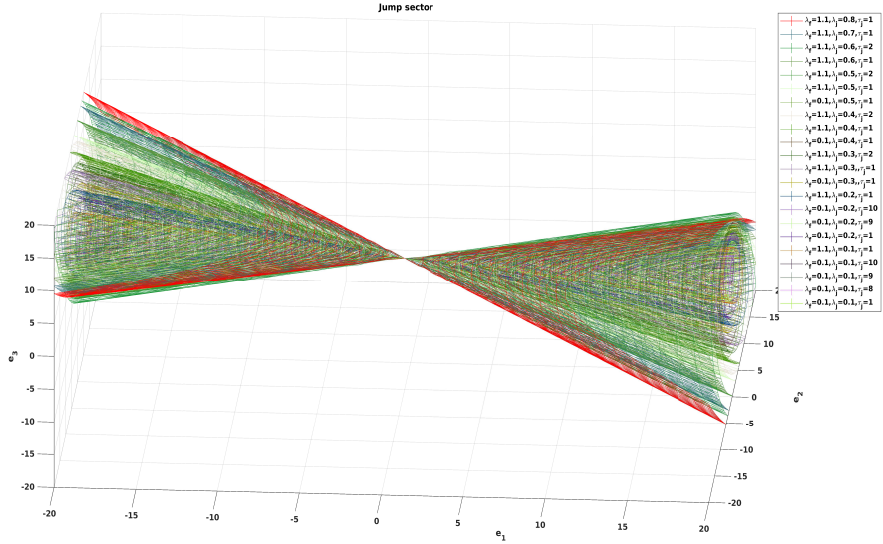}
			\caption{Jump sector}
			\label{sector}
\end{figure}

A Mont-Carlo simulation has been run for all the feasible solutions with two performance criteria  to evaluate the feasible solutions and choose the best. The first one is the $L_2$ norm of the error and the second one is the settling time (2\%) of the error. The fifth and  sixth rows of Table \ref{feas-sol} show the percentage of improvement (more than 1 percent) of the mentioned measures in the R-UIO in comparison with the C-UIO. Table \ref{feas-sol} indicates that there are 22 feasible solutions and the best solution in the sense of improving both $L_2$ norm and settling time of the error is the last entry. Therefore, the correspondent $\lambda_f=1.1$, $\lambda_j=0.8$ and $\tau_j=1$ is chosen and the simulation is continued. In this case, the related parameters will be obtained as:

		\[
			F=\begin{bmatrix}
			 -0.4090   &    0.2892     & 0.4246\\
		 	  0.2892   &     0.7555    & 0.7758\\
		 	  0.4246   &     0.7758    & 0.9560
			\end{bmatrix}
		\]
		\[
		P=\begin{bmatrix}
		   1.1029  &  -0.1262   & -0.3658\\
		    -0.1262  &  1.1057   &0.1314\\
		   -0.3658 &   0.1314   & 0.6295
		 \end{bmatrix},
	 	A_R=\begin{bmatrix}
	 	-0.0009  &  1.0000   & 0.0000\\
	 	 0.1295  &  0.3264   & 0.0000\\
	 	-0.0031 &   2.0019   & 0.0000
		 	\end{bmatrix},
 		\]	
 		 \begin{table}
 		 	\caption{choosing $\lambda_f$, $\lambda_j$ and $\tau_j$ }
 		 	\begin{adjustbox}{width=\textwidth}
 		 		\begin{tabular}{|c|c|c|c|c|c|c|c|c|c|c|c|c|c|c|c|c|c|c|c|c|c|c|}
 		 			\hline
 		 			$Feas. Sol$ & 1 & 2 & 3 & 4 & 5 & 6 & 7 & 8 & 9 & 10 & 11 & 12 & 13& 14 & 15 & 16 &17&18&19&20&21&22\\ 	\hline
 		 			$\lambda_f$ & 0.1 & 0.1 & 0.1 & 0.1  & 1.1 & 0.1 & 0.1 & 0.1  & 1.1 & 0.1 & 1.1 & 1.1 & 0.1& 1.1 &1.1&0.1&1.1&1.1&1.1&1.1&1.1&1.1\\\hline
 		 			$\lambda_j$ & 0.1 & 0.1 & 0.1 & 0.1  & 0.1 & 0.2 & 0.2 & 0.2  & 0.2 & 0.3 & 0.3 & 0.3 & 0.4& 0.4& 0.4& 0.5& 0.5& 0.5& 0.6& 0.6& 0.7& 0.8\\ \hline
 		 			$\tau_j$ & 1    &   8 &   9 & 10 &1  &  1 &  9 & 10 &1 &1 &1 &2&1&1&2&1&1&2&1&2&1&1 \\\hline
 		 			$||e||_2$& 2.3  & 0.1 & 0.1 & 0.2 &  2.7& 8.7 &  1.1&    1.1&  14.4& 12.2  & 27.6 & 19.3 & 19.3  & 		   39.8&25.8& 23.9&45.6&35.3& 53.8&42.2&62.6&67.5\\\hline	
 		 			$T_{settling}$ &6.7 &0.1 & 0.1&0.2 &19.9 &23.4 & 3.7& 1.5& 56.7 & 33.9 &  76.7 & 54.1  & 50.5 &95.9&71&56.9&98.2& 88.3&99.3&96.6&99.6&99.8 \\\hline	
 		 		\end{tabular}	\label{feas-sol}
 		 	\end{adjustbox}
 		 \end{table}
 		 \begin{remark}
 		 	The matrix $F$ should be chosen such that it is neither positive definite nor negative definite in order to represent a sector.
 		 \end{remark}
 		 	
 The next step is to define a suitable reset law which results in more improvement in the settling time and $L_2$ norm of the error. It should be reminded that in this example the state $x_2$ is not measured but its bound are known and therefore the boundary error trajectories are used. Notice that $e^TFe\leq0$ is equivalent to $\max(e_{v_1}^T(t)Fe_{v_1}(t),e_{v_2}^T(t)Fe_{v_2}(t))<0$ and we can replace this part of the condition, which is additional to the condition on $\tau_k$, with others defined in (\ref{eq:reset-laws}). Although the first one is the only reset law for which asymptotic stability is proved, it is slightly conservative thus to relax it some other reset laws are defined without stability proof.

		\begin{align} \label{eq:reset-laws}
	\left\{ 
	\begin{array}{ccccc}
			1-& \max(e_{v_1}^T(t)Fe_{v_1}(t),e_{v_2}^T(t)Fe_{v_2}(t)) & < & 0&\\
			2-& e_{v_1}^T(t)Fe_{v_1}(t)+e_{v_2}^T(t)Fe_{v_2}(t) &< &0&\\
		    3-& \max(e_{v_1}^T(t)Fe_{v_1}(t),e_{v_2}^T(t)Fe_{v_2}(t)) & < &\quad ||[e_{v_1}(t),e_{v_2}(t)]||_2e^{-t}&\\
		    4-& \left\{ \begin{array}{cccc}
				  e_{v_1}^T(t)Fe_{v_1}(t)+e_{v_2}^T(t)Fe_{v_2}(t)  \\ 
				     e_{v_1}^T(t)Fe_{v_1}(t)+e_{v_2}^T(t)Fe_{v_2}(t)  
		    	\end{array} 
		    	 \right.  
		    	 & \begin{array}{ccc}
				  &< &	  \\
				  &< & 
				 \end{array} 
				 &\begin{array}{ccc}
				  ||[e_{v_1}(t), e_{v_2}(t)]||_2\\
				 0
				\end{array} 
				&\begin{array}{ccc}
				k=1\\
				k\neq 1
				\end{array} 		    	 
	\end{array}
	\right. 
			\end{align}
			
	 The second reset law comes from the fact that the summation of the error trajectories should be inside the sector, not both of them. It means that reset may happen when only a single error trajectory is inside the sector not all of them.
			 
	 In the third reset law, the jump sector is expanded according to the norm of the error and it is exponentially decreasing to reach the first reset law. It is motivated by observing that a bigger jump sector and hence more jump probability may improve the response more. Finally, the 4th reset law is the same as the second except that the first reset may happen more quickly and leads to more improvement in the results.

	As mentioned before, since $e_2$ is not available in this example, assume that it lies in an uncertainty interval thus, use the two boundary trajectories instead of the unknown parameter and construct the reset laws based on these trajectories (Figure \ref{bound-traj}).
%	\begin{remark}
%		When only one state is not measured, the uncertainty is a line segment and reset can happen when the whole line segment is inside the jump sector (first reset law in (\ref{eq:reset-laws})). But for higher dimension cases, where two or more states are not available, the uncertainty set will be a ball, ellipse, ellipsoid or more complex sets, in these cases reset can occur whenever all the set is inside the jump sector and according to Assumption \ref{ass:convex-comb} this can be checked by boundary values of the set.   
%	\end{remark}
	
	 \begin{table}[b]
	 \caption{The effect of different reset laws on performance indices}
	  \begin{adjustbox}{width=\textwidth}
	\begin{tabular}{|c|c|c|c|c|c|c|c|c|c|c|c|c|}
		\hline
		\multirow{2}{*}{Reset law} & \multicolumn{2}{c|}{0\%-20\%} & 
		\multicolumn{2}{c|}{20\%-40\%} & 	\multicolumn{2}{c|}{40\%-60\%} & \multicolumn{2}{c|}{60\%-80\%} & \multicolumn{2}{c|}{80\%-100\%} & \multicolumn{2}{c|}{average}\\
		\cline{2-13}
	    & $||e||_2$ & $T_{stl}$ & $||e||_2$ & $T_{stl}$ & $||e||_2$ & $T_{stl}$ &$||e||_2$ & $T_{stl}$  &$||e||_2$ & $T_{stl}$ &$||e||_2$ & $T_{stl}$\\
	    	\hline
	   
	   1& 57.20 &39.30& 24.90 & 52.00 	& 8.60 & 8.30& 4.90 &0.40& 4.40&0 & 21.33 & 24.03\\\hline
	   2& 54.30 & 24.40&  23.80 & 67.20& 12.30 & 7.00& 6.00 &1.10& 3.60&0.3 & 23.19 & 25.51\\\hline
	   3&  45.00 & 1.50&  32.10 & 60.20&10.30&36.40& 7.10 &1.80&  5.30&0 &27.73  & 38.80\\\hline
	   4& 54.20 &30.40& 23.80 & 57.80&  9.90 & 9.60& 5.40 &1.20&   5.80&0.70 &  26.28& 25.90\\\hline
	\end{tabular}\label{performance-index-percent}
		 \end{adjustbox}
	\end{table}

The results of the reset laws in \eqref{eq:reset-laws} are shown in Table \ref{performance-index-percent}. In this table, the improvement of different reset laws is divided into 5 subcategories to check the improvement amount. For example in the first reset law, 57.2\% of the improvements of the $L_2$ norm are less than 20 \% and so on. The last column means that for example with the first reset condition regarding all the initial conditions, 21.33\% improvement in $||e||_2$ and also 24.03\% improvement in the settling time is achieved. Based on this table, we can see that the 3rd reset law on average is better than the others. It should be mentioned that in all cases the initial condition of the observer is zero and the initial condition of the system is a random number such that $||x||_\infty\leq20$.
%To gain a better insight about the results, Probability Density Function (PDF) of the Table \ref{performance-index-percent} is depicted in Figure \ref{fig:bar-charts-performance}. The Figure \ref{subfig:PDF_err_max_time} shows that although the 3rd reset law generally improves the $L_2$ norm of error, there are some cases which the $L_2$ norm  is degraded.

%To see the results better and for a fair comparison, starting from the same initial condition, the state estimation, sector and error trajectory and Lyapunov function are presented in Figures \ref{sts-reset-law-max} to \ref{sts-reset-law-4}. 

Using the first reset condition, the state estimation in both C-UIO and R-UIO are shown in left column while the estimation errors are shown in the right column of Figure \ref{sts-reset-law-max}. As can be seen, after the first reset the estimation error is reduced significantly. Figure \ref{JumpL-lay-error-fun-max} shows the correspondent jump sector, Lyapunov function and the root of the square error. In Figure \ref{Sec_err_traj_max} there are two error boundary trajectories associated with the two outer bounds of initial condition and the starting point is marked with a star. In the first reset condition, when both of the trajectories are inside the sector, i.e $\max(e_{v_1}^TFe_{v_1},e_{v_2}^TFe_{v_2}) \leq 0$ and also the inequality \eqref{eq:lyap-trm2} is satisfied the jump will happen. With this jump, as it is depicted in Figure \ref{lyap_max} and \ref{RSE_max} a significant decrease in the Lyapunov function  and in the square root of error is seen.

%Figure \ref{sts-reset-law-sum} and \ref{JumpL-lay-error-fun-sum} shows the result of applying the second reset condition ($ e_1^TFe_1+e_2^TFe_2  < 0 $). In this case, improvement is more than the first reset law and it seems to be less conservative.
%
%In the third reset law a time-varying threshold ($||[e_1,e_2]||_2e^{-t}$) is used to let the first reset happens quickly, but be more strict to the other jumps. As can be seen from Figures \ref{sts-reset-law-3}, \ref{JumpL-lay-error-fun-3} and also from Table \ref{performance-index-percent} the average of improvement in $||e_2||$, and in settling time is better than the previous reset laws. Therefore, apparently, it might ge good choice. But one should be aware of deterioration of $L_2$ norm in some cases as Figure \ref{subfig:PDF_err_max_time} shows.
%
%The results of final reset law are depicted in \ref{sts-reset-law-4} and \ref{JumpL-lay-error-fun-4}. Since the threshold is more restrictive than the third reset law, the improvement is less but in comparison with the first and second reset law, more improvement is achieved due to the less conservativeness of this reset law.

A quantitative comparison of these 4 reset laws with the same initial condition ($x_1=-5,x_2\in [-5,5],x_3=10$) is done in Table \ref{tble-ISE}. In this table, the first instant of the jump, $L_2$ norm and settling time (2\%) of error is calculated. As the table shows, all the proposed reset laws outperforms the C-UIO which in turn, demonstrates the effectiveness of exploiting the reset in the UIO.

%The idea of using the boundary error trajectories and reset actions is demonstrated in Figure \ref{fig:err-trajectories}. This figure shows that, if the outer error trajectories associated with two boundary initial conditions are used for designing the reset law, all other error trajectories lie between these two boundary errors and this means that, if with this reset law the observer performs better than C-UIO, it will happen for all other initial condition in the interval and therefore  the idea of using the boundary of error rather than the exact error (which is not available) is confirmed .
%\begin{figure} 
%	%			\centering
%	\includegraphics[width=\textwidth,height=12cm]{sts-err-1}
%	\caption{State estimation and estimation error}
%	\label{sts-err-1}
%\end{figure}
	%%########## reset law max (e'Fe)<0 #############################		
	\begin{figure}
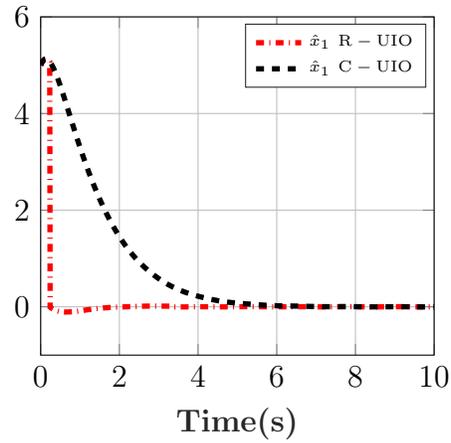
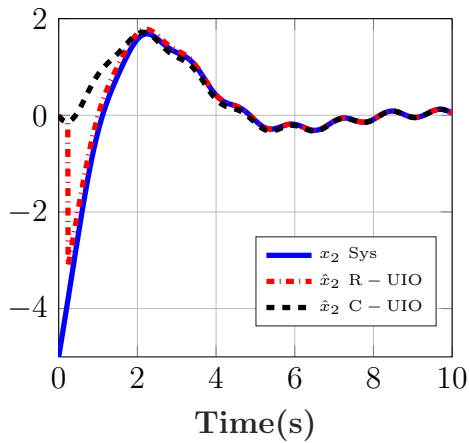
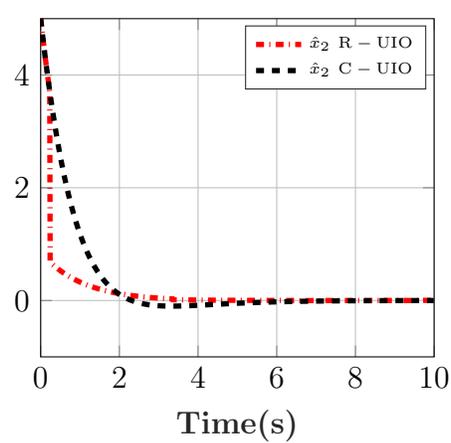
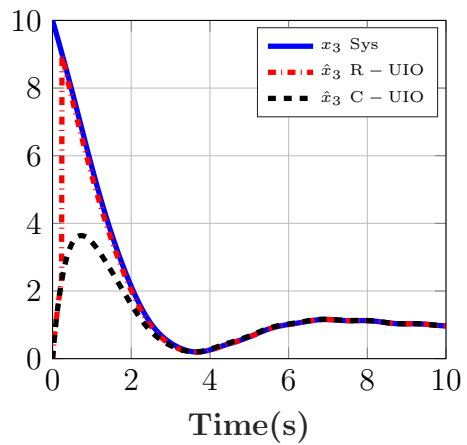
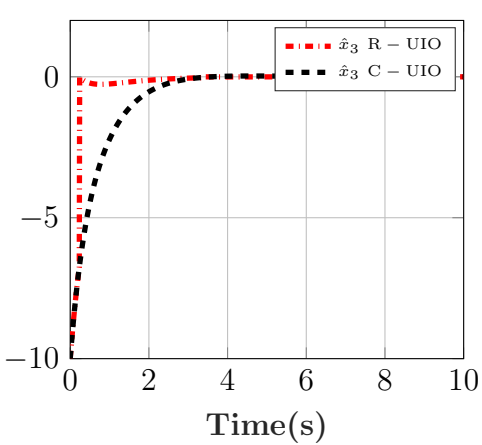

		\begin{subfigure}[b]{2.5in}
			\centering
			\setlength{\figW}{5.5cm}
			\setlength{\figH}{4.5cm}			
			\input{X1.tikz}
			\caption{first state of the plant}
		\end{subfigure}
		\begin{subfigure}[b]{2.5in}
			\centering 
			\setlength{\figW}{5.5cm}
			\setlength{\figH}{4.5cm}			
			\input{X1_err.tikz}
			\caption{first state estimation error}		
		\end{subfigure}	\\
		\begin{subfigure}[b]{2.5in}
			\vspace{.2cm}
			\centering
			\setlength{\figW}{5.5cm}
			\setlength{\figH}{4.5cm}			
			\input{X2.tikz}
			\caption{second state of the plant}
		\end{subfigure}
		\begin{subfigure}[b]{2.5in}
			\centering
				\setlength{\figW}{5.5cm}
				\setlength{\figH}{4.5cm}			
				\input{X2_err.tikz}
			\caption{second state estimation error}
			
		\end{subfigure}	
		
		\begin{subfigure}[b]{2.5in}
			\vspace{.2cm}
			\centering
				\setlength{\figW}{5.5cm}
				\setlength{\figH}{4.5cm}			
				\input{X3.tikz}
			\caption{third state of the plant}
			%			\vspace{.2cm}
		\end{subfigure}
		\begin{subfigure}[b]{2.5in}
			\centering
		\setlength{\figW}{5.5cm}
		\setlength{\figH}{4.5cm}			
		\input{X3_err.tikz}
			\caption{third state estimation error}
		\end{subfigure}			
		\caption{State estimation (left column) and its error  (right column) with the first reset law}
		\label{sts-reset-law-max} 
	\end{figure} 
	
	\begin{figure} 
		\begin{subfigure}[b]{4.5in}
			%			\centering
			\includegraphics[width=9.4cm,center]{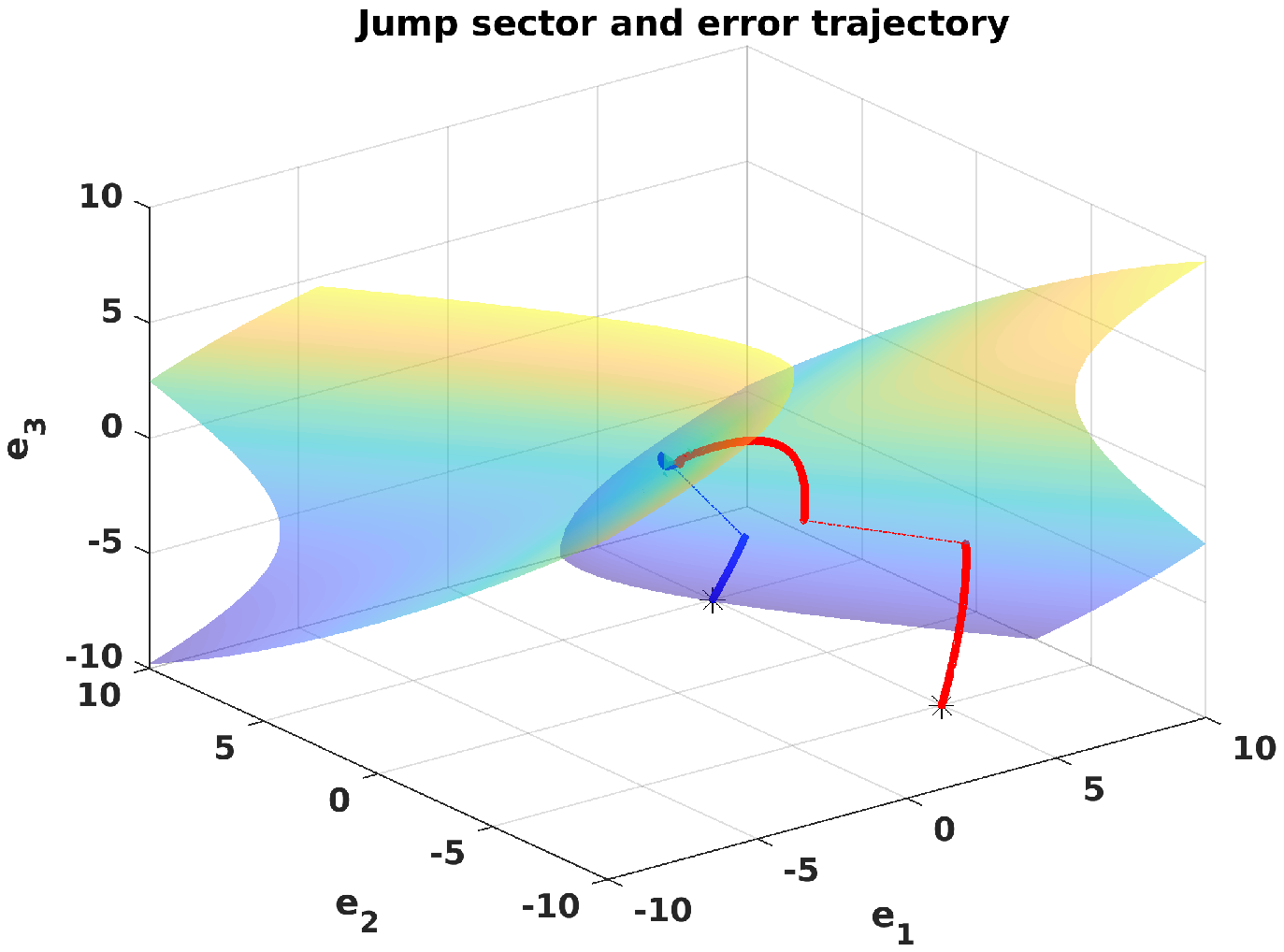}
			\caption{Jump sector and error trajectory}
			\label{Sec_err_traj_max}
		\end{subfigure}\\
		\begin{subfigure}[t]{3in}
			%			\centering
				\setlength{\figW}{6.5cm}
				\setlength{\figH}{5cm}
				
				\input{lyap2nd.tikz}
			\caption{Lyapunov Function}
			\label{lyap_max}
		\end{subfigure}
		\begin{subfigure}[t]{3in}
			%			\centering
			\setlength{\figW}{6.5cm}
			\setlength{\figH}{5cm}
			\input{RSE2nd.tikz}
			\caption{Root Square Error }
			\label{RSE_max}
		\end{subfigure}
		\caption{Jump sector, Lyapunov function and error function}
		\label{JumpL-lay-error-fun-max}
	\end{figure}

			\begin{table}
				\caption{Comparison of different reset laws}
				\label{tble-ISE}
				\centering
				 \begin{adjustbox}{width=\textwidth}
				\begin{tabular}{|c|c|c|c|c|c|}				
					\hline
					  &1st reset law & 2nd reset law & 3rd reset law & 4th reset law&  Conv-UIO\\
					\hline
					$T_{1stReset}(s)$&0.236 &  0.182 & 0.193& 0.148  & - \\
						\hline
					$\sqrt{\int_{0}^{\infty}{e^Tedt}}$ & 5.1992  &  4.8050  & 5.1783 & 4.3270 &8.1944\\
					\hline
					Settling time (2\%) & 3.356 & 3.106 &4.229&  3.071 & 6.715\\
					\hline
					
				\end{tabular}
			\end{adjustbox}	
			\end{table}
	
%				\begin{figure}[t]
%								\centering
%					\includegraphics[width=11cm]{diff-init}
%					\caption{error trajectories}
%					\label{fig:err-trajectories}
%				\end{figure}
	%############################################################
	\section{Conclusion} \label{sec:conclusion}
	In this paper, Reset Unknown Input Observer is proposed in which the states of the observer reset to a suitable value based on a time-dependent reset law. Design starts with an ideal case and a jump sector is obtained. Then, in non-ideal case we used the boundary error trajectories to determine the reset times. Moreover, we analyzed the stability and convergence to show that the estimation error will converge to zero asymptotically. Furthermore,  we exploited a simulation example to demonstrates the efficiency of using the reset in the UIO to decrease the $L_2$ and settling time of estimation error. Moreover, to relax the conservatism of the proposed reset law, we presented some other reset laws. Although such reset laws may perform nicely in some cases, there is no rigorous stability proof for them.  The focus of our future work is on developing less conservative reset laws with stability proof using  the  presented R-UIO. 
		%############################################################
%	\clearpage

\end{document}